%% file: 2026_06_16_elsarticle_main_V6.tex
\documentclass[final,5p,times,twocolumn]{elsarticle}
\journal{Systems \& Control Letters}

\usepackage[utf8]{inputenc}
\usepackage{textcomp}
\usepackage{amssymb}
\usepackage{amsmath}
\usepackage{graphicx}
\usepackage[version=4]{mhchem}
\usepackage{siunitx}
\usepackage{longtable,tabularx}
\usepackage{algorithm}
\usepackage{algpseudocode}

\input{PackagesCommands}

\begin{document}

\begin{frontmatter}

\title{Exactness Certificates for Closed-Form CBF Safety-Filter Projections}

\author[umbc]{Ankit Goel\corref{cor1}}
\ead{ankgoel@umbc.edu}

\cortext[cor1]{Corresponding author.}

\affiliation[umbc]{
    organization={Department of Mechanical Engineering, University of Maryland, Baltimore County},
    addressline={1000 Hilltop Circle},
    city={Baltimore},
    postcode={21250},
    state={MD},
    country={USA}
}

\begin{abstract}
For control-affine systems, standard and high-order control barrier function conditions are affine in the control input and are commonly enforced through quadratic-program-based safety filters.
Although convex, these optimization problems may be undesirable in embedded, high-rate, or resource-limited implementations.
This letter characterizes when the corresponding Euclidean projection can be recovered from the affine inequalities violated by a nominal control input.
Given a nominal input, we form the violated set and compute the minimum-norm correction that enforces the violated inequalities with equality.
This violated-set correction is closed form, but it need not equal the exact Euclidean projection onto the full feasible set.
The main result gives a necessary and sufficient exactness certificate based on primal and dual feasibility, followed by structural sufficient conditions involving interactions among affine-inequality normals.
An online certification algorithm is then presented to determine when the closed-form update is exact.
When the certificate fails, a finite active-set search can be used to recover the exact projection.
Numerical simulations illustrate that the violated-set correction can remain feasible while failing to be the exact projection due to dual infeasibility, and demonstrate computational speedup relative to a standard CBF-QP solver.
\end{abstract}


\end{frontmatter}



\section{Introduction}

Control barrier functions (CBFs) provide a systematic approach for enforcing safety constraints in dynamic systems \cite{ames2019control}. 
In control-affine systems, with standard CBFs and high-order CBFs, the resulting safety condition is affine in the control input.
Thus, at each time instant, the safety filter is commonly implemented by solving a quadratic program that minimally modifies a nominal control input while enforcing the affine CBF inequalities \cite{ames2014control,ames2016control}.
CBF-QP safety filters have been widely used as real-time interventions that minimally modify nominal control inputs to enforce safety constraints \cite{ames2016control,wabersich2021predictive}.

Although CBF-QPs are convex, solving an optimization problem at every time step can be undesirable in embedded, high-rate, or resource-limited control applications.
This motivates the development of closed-form safety filters that preserve the Euclidean projection interpretation of the CBF-QP while avoiding a numerical optimizer whenever possible.

This letter characterizes one such construction.
Given a nominal control input, we identify the affine inequalities violated by that nominal
input and compute the minimum-norm correction that enforces those inequalities with equality. 
In general, this violated-set correction need not coincide with the exact Euclidean projection onto the feasible set.
The main contribution of this letter is to characterize when it does. 

Specifically, we derive a necessary and sufficient exactness certificate based on primal and dual feasibility, provide a structural sufficient condition involving interactions between affine-inequality normals, and obtain geometric special cases.
We also present an online certification procedure for determining when the closed-form update is exact, together with a finite active-set search for recovering the exact projection when the certificate fails.

Related work on closed-form and explicit CBF safety filters has grown recently.
A closed-form expression for the CBF-QP solution is explored in~\cite{mestres2025explicit} by partitioning the state space into regions, each admitting a distinct closed-form optimizer, and using this partition offline to avoid calling a numerical solver between region changes. 
The present letter differs in that no offline partitioning is required: given a nominal input at any state, the violated-set correction is computed online and immediately certified through primal and dual feasibility conditions. 
Compatibility of multiple CBF constraints under input bounds is studied in~\cite{tan2022compatibility} by establishing conditions under which the CBF-QP is well-defined when multiple constraints are simultaneously active.
The present work instead takes feasibility as given and characterizes when the violated-set correction coincides with the exact Euclidean projection.

The remainder of this letter is organized as follows.
{Section~\ref{sec:cbf_projection_filter} formulates the projection-based CBF safety filter and introduces the violated-set correction.}
{Section~\ref{sec:main_result} presents the exactness certificate, structural sufficient condition, and geometric special cases.}
{Section~\ref{sec:algorithm} presents the online certification algorithm.}
{Section~\ref{sec:active_set_search} presents the finite active-set search for recovering the exact projection when the certificate fails.}
{Section~\ref{sec:numerical_simulation} illustrates the certificate and active-set recovery mechanism in simulation, and Section~\ref{sec:conclusion} concludes the letter.}


\section{Projection-Based CBF Safety Filters and Violated-Set Corrections}
\label{sec:cbf_projection_filter}

This section formulates the projection problem induced by affine CBF constraints and introduces the closed-form violated-set correction studied in the remainder of the paper.
For control-affine systems, standard CBF and HOCBF conditions are affine in the control input, and the corresponding safety filter is commonly implemented as the Euclidean projection of a nominal input onto a polyhedral feasible set.
The question addressed here is whether this projection can be recovered directly from the affine inequalities that are violated by the nominal input.
The resulting violated-set correction is closed form, but it is not always equal to the exact Euclidean projection onto the full feasible set. 
The purpose of the following formulation is to define this candidate precisely and set up the exactness certificates developed in the next section.

Consider the control-affine system
\begin{align}
    \dot x
        =
            f(x) + g(x)u,
\end{align}
where $x\in \BBR^n$ is the state and $u \in \BBR^m$ is the control input.
For systems that are affine in the control input, standard CBF and HOCBF conditions are also affine in the control input \cite{xiao2019control,xiao2021high}. 
Therefore, the resulting safety condition can be written in the linear inequality form
\begin{align}
    A(x)u
        \le
            b(x),
    \label{eq:cbf_constraint}
\end{align}
where
$A(x) \in \BBR^{q\times m}$
and 
$b(x) \in \BBR^{q}.$

Note that the matrices $A(x)$ and $b(x)$ depend on the chosen CBF or HOCBF construction and on the system dynamics.
Thus, the conditions derived below are conditions on the control-space affine inequalities induced by the state constraints, rather than direct geometric conditions on the desired safe set.

Let $u_{\rm nom}$ denote a nominal control input.
The standard CBF safety filter is the Euclidean projection
\begin{align}
    u_{\rm safe}
        =
            \argmin_u
            \frac12
            \|u-u_{\rm nom}\|_2^2
\end{align}
subject to \eqref{eq:cbf_constraint}.
This is the Euclidean projection of $u_{\rm nom}$ onto the polyhedral feasible set defined by the affine CBF inequalities.




Since the results below depend only on the affine inequalities in control space, we suppress the dependence on $x$ and consider the abstract polyhedral set
\begin{align}
    \SU
        \isdef
            \left\{
                u \in \real^m
                :
                Au \le b
            \right\},
    \label{eq:feasibleSet}
\end{align}
where
$A
        \in
            \real^{q\times m},$
and 
$b
        \in
            \real^q.$
Let
$u_{\rm nom}
        \in
            \real^m$
denote a nominal control input that may not satisfy the constraints defining $\SU$.
%
The Euclidean projection of $u_{\rm nom}$ onto $\SU$ is
\begin{align}
    u^*
        \isdef
            \argmin_{u \in \SU}
            \frac{1}{2}
            \|
                u-u_{\rm nom}
            \|_2^2.
    \label{eq:projection_qp}
\end{align}
The safety filter is thus a Euclidean projection onto a polyhedral feasible set, a standard convex quadratic program \cite{boyd2004convex}.

Define the constraint residual vector
\begin{align}
    r
        \isdef
            A u_{\rm nom}
            -
            b,
\end{align}
the violated constraint set
\begin{align}
    \SV
        \isdef
            \left\{
                i \in {1,\ldots,q}
                :
                r_i > 0
            \right\}.
\end{align}
and the complement of the violated constraint set 
\begin{align}
    \SV^c
        \isdef
            \{1,\ldots,q\}
            \setminus
            \SV.
\end{align}
Let $A_{\SV}$ and $A_{\SV^c}$ denote the submatrices formed by the rows of $A$ corresponding to the indices in $\SV$ and $\SV^c$, respectively. Likewise, let $b_{\SV}$, $b_{\SV^c}$, $r_{\SV}$, and $r_{\SV^c}$ denote the corresponding subvectors.


Note that if \(\SV=\emptyset\), then \(u_{\rm nom}\in\SU\) and \(u^*=u_{\rm nom}\). Hence, the nontrivial case is \(\SV\neq\emptyset\).
Next, assume that $A_{\SV}$ has full row rank.
The violated-set correction is defined by projecting $u_{\rm nom}$ onto the affine subspace obtained by enforcing the violated inequalities as equalities, that is, 
\begin{align}
    u_{\SV}
        &\isdef
            u_{\rm nom}
            -
            A_{\SV}^T
            \left(
                A_{\SV}A_{\SV}^T
            \right)^{-1}
            \left(
                A_{\SV}u_{\rm nom}
                -
                b_{\SV}
            \right).
    \label{eq:violated_projection}
\end{align}
By construction, $u_{\SV}$ is the orthogonal projection of $u_{\rm nom}$ onto the affine set
\begin{align}
A_{\SV}u = b_{\SV}.
\end{align}
However, this affine set is generally not equal to the full feasible set $\SU$.
Therefore, $u_{\SV}$ need not coincide with the Euclidean projection $u^*$.

The objective of this paper is to determine conditions under which the violated-set correction coincides with the exact Euclidean projection, that is,
\begin{align}
    u_{\SV}
        =
            u^*.
\end{align}

When this condition holds, the CBF-QP projection can be recovered from the currently violated affine inequalities.
When it fails, the violated-set correction is not certified to be the closest feasible input, and 
an exact recovery step, such as the finite active-set search in Section~\ref{sec:active_set_search}, is required.

\section{Exactness of the Violated-Set Correction}
\label{sec:main_result}

This section characterizes when the violated-set correction defined in Section~\ref{sec:cbf_projection_filter} coincides with the exact Euclidean projection onto the full feasible set.
The key point is that the affine inequalities violated by $u_{\rm nom}$ need not coincide with the active constraints of the projection problem.
Consequently, enforcing the violated inequalities as equalities may produce a candidate that is either infeasible with respect to previously satisfied inequalities or feasible but not closest to $u_{\rm nom}$. 
The results below distinguish these cases through primal and dual feasibility.

\subsection{Exactness Certificate}

The following result gives a necessary and sufficient condition for the violated-set correction to coincide with the Euclidean projection onto the full feasible set.

The full-row-rank assumption on \(A_{\SV}\) in the following result is a nondegeneracy condition on the affine CBF inequalities in control space.
It requires the violated constraint normals induced by the CBF conditions to be linearly independent. 
This condition is distinct from the geometry of the state-space safe set; for example, a box-shaped safe set in the state variables need not induce a box-shaped feasible set in the control input.

\begin{theorem}[Exactness Certificate]
\label{thm:exactness_characterization}
Assume that $A_{\SV}$ has full row rank.
Then $u_{\SV}$, given by \eqref{eq:violated_projection}, is the Euclidean projection of $u_{\rm nom}$ onto $\SU$ if and only if
\begin{align}
    Au_{\SV}
        \le
            b
    \label{eq:exactness_primal_feasibility}
\end{align}
and
\begin{align}
    \left(
        A_{\SV}A_{\SV}^T
    \right)^{-1}
    \left(
        A_{\SV}u_{\rm nom}
        -
        b_{\SV}
    \right)
        \ge
            0.
    \label{eq:exactness_dual_feasibility}
\end{align}
\end{theorem}

\begin{nonumberproof}
Since $A_{\SV}$ has full row rank,
\eqref{eq:violated_projection} is the orthogonal projection of $u_{\rm nom}$ onto the affine set $A_{\SV}u = b_{\SV}.$
Therefore, $A_{\SV}u_{\SV} = b_{\SV}.$

Define
\begin{align}
    \lambda_{\SV}
        \isdef
            \left(
                A_{\SV}A_{\SV}^T
            \right)^{-1}
            \left(
                A_{\SV}u_{\rm nom}
                -
                b_{\SV}
            \right),
\end{align}
and, for $i \notin \SV,$ set $\lambda_i = 0.$
Then, \eqref{eq:violated_projection} implies
\begin{align}
    u_{\SV}
        -
        u_{\rm nom}
        +
        A^T\lambda
        =
            0.
\end{align}
Thus, stationarity holds.
Moreover, complementarity holds because
$A_{\SV}u_{\SV}=b_{\SV}$ and $\lambda_i=0$ for all $i\notin\SV$.
Therefore, the KKT conditions hold if and only if
\eqref{eq:exactness_primal_feasibility} and
\eqref{eq:exactness_dual_feasibility} hold.
Since the objective in \eqref{eq:projection_qp} is strictly convex and $\SU$ is convex, these conditions are necessary and sufficient for $u_{\SV}=u^\ast$.
\end{nonumberproof}

Theorem~\ref{thm:exactness_characterization} separates the two ways in which the violated-set correction can fail to be exact. 
Condition~\eqref{eq:exactness_primal_feasibility} checks whether enforcing the violated inequalities has made any previously satisfied inequality infeasible.
Condition~\eqref{eq:exactness_dual_feasibility} checks whether the violated inequalities can serve as active constraints of the Euclidean projection with nonnegative KKT multipliers. 
Thus, a candidate can be feasible and still fail to be the closest feasible input if the associated multiplier vector is not componentwise nonnegative.

\subsection{Structural Sufficient Condition}

The exactness certificate in Theorem~\ref{thm:exactness_characterization} can be checked directly after computing the violated-set correction. 
We next give a sufficient condition that guarantees this certificate using only algebraic interactions among the affine-inequality normals.
The first condition ensures nonnegativity of the multiplier candidate associated with the violated constraints.
The second condition ensures that the violated-set correction does not increase the residuals of constraints that were already satisfied by \(u_{\rm nom}\).

\begin{theorem}[Structural Exactness Condition]
\label{thm:exactness}
Assume that $A_{\SV}$ has full row rank.
Furthermore, assume that
\begin{align}
    \left(
        A_{\SV}A_{\SV}^T
    \right)^{-1}
        \ge
            0,
    \label{eq:inverse_positive_condition}
\end{align}
and
\begin{align}
    A_{\SV^c}A_{\SV}^T
    \left(
        A_{\SV}A_{\SV}^T
    \right)^{-1}
        \ge
            0,
    \label{eq:feasibility_preservation_condition}
\end{align}
where the inequalities are interpreted componentwise.
Then,
\begin{align}
    u_{\SV}
        =
            u^*.
    \label{eq:main_result}
\end{align}

\end{theorem}

\begin{proof}
Since $\SV$ is the violated constraint set, it follows that 
$r_{\SV} > 0,$ 
and 
$r_{\SV^c} \le 0,$
where the inequalities are interpreted componentwise.

Using $u_\SV$ given by \eqref{eq:violated_projection}, it follows that 
\begin{align}
    A_{\SV}u_{\SV}-b_{\SV}
        &=
            r_{\SV}
            -
            A_{\SV}A_{\SV}^T
            \left(
                A_{\SV}A_{\SV}^T
            \right)^{-1}
            r_{\SV}
       =
            0 .
    \label{eq:violated_residual_after_projection}
\end{align}
For the remaining constraints,
\begin{align}
    A_{\SV^c}u_{\SV}-b_{\SV^c}
        &=
            r_{\SV^c}
            -
            A_{\SV^c}A_{\SV}^T
            \left(
                A_{\SV}A_{\SV}^T
            \right)^{-1}
            r_{\SV}.
    \nn
\end{align}
Using \eqref{eq:feasibility_preservation_condition} and the fact that  $r_{\SV}>0$, it follows that 
\begin{align}
    A_{\SV^c}A_{\SV}^T
    \left(
        A_{\SV}A_{\SV}^T
    \right)^{-1}
    r_{\SV}
        \ge
            0.
    \nn
\end{align}
Therefore,
\begin{align}
    A_{\SV^c}u_{\SV}-b_{\SV^c}
        \le
            r_{\SV^c}
        \le
            0 .
    \nn
\end{align}
Together with \eqref{eq:violated_residual_after_projection}, this gives
\begin{align}
    Au_{\SV}
        \le
            b .
    \nn
\end{align}
Thus, $ u_{\SV}$  is feasible.

Next, define
\begin{align}
    \lambda_{\SV}
        \isdef
            \left(
                A_{\SV}A_{\SV}^T
            \right)^{-1}
            r_{\SV}.
    \nn
\end{align}
By \eqref{eq:inverse_positive_condition} and $ r_{\SV}>0$, it follows that 
$\lambda_{\SV} \ge 0 .$
Hence, the primal- and dual-feasibility conditions of Theorem~\ref{thm:exactness_characterization} hold, which implies \eqref{eq:main_result}. 
\end{proof}

Theorem~\ref{thm:exactness} separates the two roles of constraint-normal interactions. Condition~\eqref{eq:inverse_positive_condition} guarantees dual feasibility of the multiplier candidate for every positive violated residual vector \(r_{\SV}\). Condition~\eqref{eq:feasibility_preservation_condition} guarantees primal feasibility by ensuring that the correction associated with the violated constraints does not increase the residuals of the constraints that were initially satisfied. 
These conditions are sufficient, but not necessary, for exactness.


\subsection{Geometric Special Cases}

The structural conditions in Theorem~\ref{thm:exactness} can be interpreted through interactions among the affine-inequality normals. 
This subsection presents two simple cases in which these interactions guarantee exactness.
The first corresponds to mutually orthogonal constraint normals, for which the violated-set correction does not alter the residuals of the remaining constraints. 
The second corresponds to a single violated affine inequality, where exactness depends only on the sign of its interaction with the constraints that are already satisfied by $u_{\rm nom}$.


\begin{corollary}[Orthogonal Constraint Normals]
\label{cor:orthogonal}
Suppose that the rows of $A$ are mutually orthogonal, that is, for $i \neq j,$ $A_iA_j^T = 0.$
Then
\begin{align}
    u_{\SV}
        =
            u^*.
\end{align}
\end{corollary}

\begin{nonumberproof}
Since the rows of $A$ are mutually orthogonal,
$A_{\SV}A_{\SV}^T$ is diagonal. 
Hence,
$\left(
        A_{\SV}A_{\SV}^T
    \right)^{-1}$
is diagonal with strictly positive diagonal entries and
therefore
$\left(
        A_{\SV}A_{\SV}^T
    \right)^{-1}
        \ge
            0.$

Moreover,
$A_{\SV^c}A_{\SV}^T
        =
            0.$
Therefore,
\begin{align}
    A_{\SV^c}A_{\SV}^T
    \left(
        A_{\SV}A_{\SV}^T
    \right)^{-1}
        =
            0
        \ge
            0.
\end{align}
The result then follows directly from Theorem~\ref{thm:exactness}.
\end{nonumberproof}

Thus, when the normals associated with the affine inequalities violated by $u_{\rm nom}$ are orthogonal to the normals associated with the remaining affine inequalities, the correction needed to enforce one such inequality does not change the residual of any other affine inequality.

\begin{corollary}[Single Violated Constraint]
\label{cor:single_violated_constraint}
Suppose that the nominal control violates only the $ i$ th affine
inequality, that is,
$\SV
        =
            \{i\}.$
If
$A_{\SV^c}A_i^T
        \ge
            0,$
then
\begin{align}
    u_{\SV}
        =
            u^*.
\end{align}
\end{corollary}

\begin{nonumberproof}
Since $\SV=\{i\}$,
$A_{\SV}A_{\SV}^T
        =
            A_iA_i^T > 0.$
Thus,
\begin{align}
    \left(
        A_{\SV}A_{\SV}^T
    \right)^{-1}
        =
            \frac{1}{A_iA_i^T}
        >
            0.
\end{align}
Moreover,
\begin{align}
    A_{\SV^c}A_{\SV}^T
    \left(
        A_{\SV}A_{\SV}^T
    \right)^{-1}
        =
            \frac{
                A_{\SV^c}A_i^T
            }{
                A_iA_i^T
            }
        \ge
            0.
\end{align}
Therefore, the conditions of Theorem~\ref{thm:exactness}
hold, and hence $u_{\SV}=u^*$.
\end{nonumberproof}

The condition $ A_{\SV^c}A_i^T\ge 0$  means that each affine inequality satisfied by $ u_{\rm nom}$  has a nonnegative interaction with the affine inequality violated by $ u_{\rm nom}$. 
Therefore, the correction that enforces the $ i$ th affine inequality with equality does not increase the residual of any affine inequality that was already satisfied.

Figure~\ref{fig:single_positive_interaction} shows a configuration in which the nominal control violates only one affine inequality.
Since the interaction between the violated affine inequality and the previously satisfied affine inequality is positive, the violated-set correction moves the control input in a direction that decreases the residual of the previously satisfied inequality. 
Hence, the corrected input remains feasible and coincides with the exact Euclidean projection.

\begin{figure}[h]
\centering
\begin{tikzpicture}[scale=1.0, >=latex]

    \draw[->, gray] (-0.5,0) -- (4.5,0) node[right] {$u_1$};
    \draw[->, gray] (0,-0.5) -- (0,4.0) node[above] {$u_2$};

    \draw[line width=1.1] (2,-0.2) -- (2,3.7)
        node[above] {$A_1u=b_1$};

    \draw[line width=1.1] (0.4,3.6) -- (4.3,-0.3)
        node[right] {$A_2u=b_2$};

    \fill[blue!12] (-0.2,-0.2) -- (2,-0.2) -- (2,2) -- (0.4,3.6) -- (-0.2,3.6) -- cycle;
    \node at (0.85,1.4) {$\SU$};

    \coordinate (unom) at (3.1,0.6);
    \coordinate (uV) at (2,0.6);

    \filldraw[black] (unom) circle (2pt)
        node[below] {$u_{\rm nom}$};

    \filldraw[black] (uV) circle (2pt)
        node[below left] {$u_{\SV}=u^*$};

    \draw[->, line width=1.0] (unom) -- (uV);

    \draw[->, thick] (2.25,2.8) -- +(0.6,0)
        node[right] {$A_1^T$};

    \draw[->, thick] (2.9,1.6) -- +(0.45,0.45)
        node[right] {$A_2^T$};

    \node[align=center] at (3.35,3.25)
        {$A_2A_1^T>0$};

\end{tikzpicture}
\caption{Single violated affine inequality with positive interaction.
The violated-set correction decreases the residual of the nonviolated
affine inequality, and hence $u_{\SV}=u^*$.}
\label{fig:single_positive_interaction}
\end{figure}
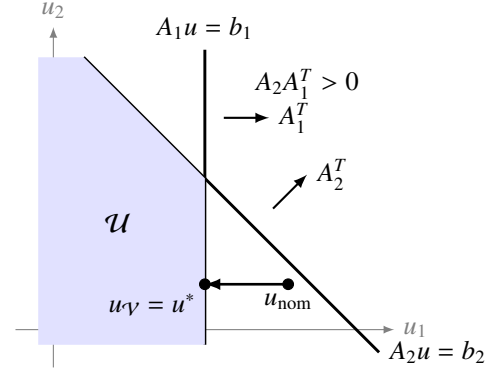


Figure~\ref{fig:single_negative_interaction} shows a configuration in which the nominal control again violates only one affine inequality, but the interaction with a previously satisfied affine inequality is negative.
In this case, the violated-set correction increases the residual of the previously satisfied
inequality, so the corrected input is not feasible and does not coincide with the exact Euclidean projection. 
In the configuration shown, the exact projection occurs at the intersection of the two active affine inequalities, although in general it may occur on a different active face of the feasible set.

\begin{figure}[h]
\centering
\begin{tikzpicture}[scale=1.0, >=latex]

    \draw[->, gray] (-0.5,0) -- (4.5,0) node[right] {$u_1$};
    \draw[->, gray] (0,-0.5) -- (0,4.0) node[above] {$u_2$};

    \draw[line width=1.1] (2,-0.2) -- (2,3.7)
        node[above] {$A_1u=b_1$};

    \draw[line width=1.1] (-0.2,0.0) -- (3.6,3.8)
        node[right] {$A_2u=b_2$};

    \fill[blue!12] (-0.2,-0.2) -- (2,-0.2) -- (2,2.2) -- (-0.2,0.0) -- cycle;
    \node at (0.85,0.65) {$\SU$};

    \coordinate (unom) at (3.1,2.0);
    \coordinate (uV) at (2,2.0);
    \coordinate (ustar) at (2,2.2);

    \filldraw[black] (unom) circle (2pt)
        node[below right] {$u_{\rm nom}$};

    \filldraw[black] (uV) circle (2pt)
        node[below right] {$u_{\SV}$};

    \filldraw[black] (ustar) circle (2pt)
        node[above left] {$u^*$};

    \draw[->, line width=1.0] (unom) -- (uV);
    \draw[->, dashed, line width=1.0] (unom) -- (ustar);

    \draw[->, thick] (2.25,0.65) -- +(0.6,0)
        node[right] {$A_1^T$};

    \draw[->, thick] (1.0,2.0) -- +(-0.45,0.45)
        node[left] {$A_2^T$};

    \node[align=center, right] at (3.2,3.25)
        {$A_2A_1^T<0$};

\end{tikzpicture}
\caption{Single violated affine inequality with negative interaction.
The violated-set correction increases the residual of a previously
satisfied affine inequality, so $u_{\SV}$ is not feasible and
$u_{\SV}\neq u^*$.}
\label{fig:single_negative_interaction}
\end{figure}
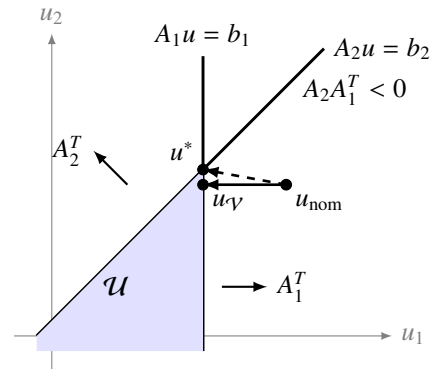

\section{Online Certification Algorithm}
\label{sec:algorithm}

The preceding results lead to the online safety-filter implementation summarized in Algorithm~\ref{alg:optimization_free_cbf_filter}. 
At each state $x$, the CBF or HOCBF conditions define affine inequalities of the form \eqref{eq:cbf_constraint}.
Given a nominal input $u_{\rm nom}$, the filter first identifies the affine inequalities that are violated by $u_{\rm nom}$. 
It then computes the minimum-norm correction that enforces only those inequalities with
equality.
The candidate is accepted only if it satisfies the primal- and dual-feasibility
conditions of Theorem~\ref{thm:exactness_characterization}. 
If the certificate holds, then the candidate is the exact Euclidean projection onto the full feasible set, and the CBF-QP solution has been recovered in closed form.

If the certificate fails, then the violated-set correction is not certified to be the closest feasible input. 
In this case, an exact recovery step must be used to recover the exact projection, such as the finite active-set search described in Section~\ref{sec:active_set_search} or a standard CBF-QP.
Thus, the algorithm bypasses numerical optimization only when the closed-form candidate is certified to be exact.


\begin{algorithm}[h]
\caption{Certified Closed-Form CBF Safety Filter}
\label{alg:optimization_free_cbf_filter}
\begin{algorithmic}[1]
\Require State $x$, nominal input $u_{\rm nom}$, affine CBF inequality data $A(x)$ and $b(x)$
\Ensure Filtered input $u_{\rm safe}$

\State Compute the residual
\begin{align}
    r
        =
            A(x)u_{\rm nom}-b(x). \nn
\end{align}

\State Determine the index set of affine inequalities violated by $u_{\rm nom}$:
\begin{align}
    \SV
        =
            \{i\in\{1,\ldots,q\}: r_i>0\}. \nn
\end{align}

\If{$\SV=\emptyset$}
    \State Set $u_{\rm safe}=u_{\rm nom}$.
\Else
    \State Form $A_{\SV}(x)$, $b_{\SV}(x)$, and $r_{\SV}$.
     \If{\(A_{\SV}(x)\) does not have full row rank}
        \State Use an exact recovery step, such as Algorithm~\ref{alg:active_search} or a standard CBF-QP.
    \Else
    \State Compute the violated-set correction
    \begin{align}
        u_{\SV}
            =
                u_{\rm nom}
                -
                A_{\SV}(x)^T
                \left(
                    A_{\SV}(x)A_{\SV}(x)^T
                \right)^{-1}
                r_{\SV}. \nn
    \end{align}

    \State Compute the multiplier candidate
    \begin{align}
        \lambda_{\SV}
            =
                \left(
                    A_{\SV}(x)A_{\SV}(x)^T
                \right)^{-1}
                r_{\SV}. \nn
    \end{align}

    \If{$A(x)u_{\SV}\le b(x)$ and $\lambda_{\SV}\ge 0$}
        \State Set $u_{\rm safe}=u_{\SV}$.
    \Else
        \State Use an exact recovery step, such as Algorithm~\ref{alg:active_search} or a standard CBF-QP.
    \EndIf
    \EndIf
\EndIf
\end{algorithmic}
\end{algorithm}

\newpage
\section{Finite Active-Set Search for Exact Projection}
\label{sec:active_set_search}

As shown in the preceding section, the violated-set correction may fail the exactness certificate.
When this occurs, the exact Euclidean projection can still be recovered without invoking a numerical quadratic-program solver by using an algebraic active-set search.  
The idea is to enumerate candidate active sets, compute the corresponding equality projection in closed form, and accept a candidate only when it satisfies the same primal-dual KKT certificate used in Algorithm~\ref{alg:optimization_free_cbf_filter}.

Specifically, let $\mathcal I\subseteq\{1,\ldots,q\}$ denote a candidate active set.
Since the projection is computed in the $m$-dimensional control space, at most $m$ linearly independent affine constraints can be active in a nonredundant active-set representation.  
Thus, even if the projected point lies on more than $m$ constraint boundaries, a linearly independent subset of at most $m$ active constraints is sufficient to certify the same projected input through the KKT conditions.
Therefore, it is sufficient to consider candidate active sets satisfying
\begin{align}
    |\mathcal I|
        \le
            \min\{m,q\}.
\end{align}

For each such set $\mathcal I$, if $A_{\mathcal I}$ has full row rank, define
\begin{align}
    u_{\mathcal I}
        =
            u_{\rm nom}
            -
            A_{\mathcal I}^T
            \left(
                A_{\mathcal I}A_{\mathcal I}^T
            \right)^{-1}
            \left(
                A_{\mathcal I}u_{\rm nom}
                -
                b_{\mathcal I}
            \right),
\end{align}
and the corresponding multiplier candidate
\begin{align}
    \lambda_{\mathcal I}
        =
            \left(
                A_{\mathcal I}A_{\mathcal I}^T
            \right)^{-1}
            \left(
                A_{\mathcal I}u_{\rm nom}
                -
                b_{\mathcal I}
            \right).
\end{align}
The candidate is accepted if $Au_{\mathcal I}\le b,$ and $\lambda_{\mathcal I}\ge 0.$
In that case, the KKT conditions are satisfied, with zero multipliers assigned to all constraints outside $\mathcal I$, and $u_{\mathcal I}$ is the exact Euclidean projection of $u_{\rm nom}$ onto the feasible set.

\begin{algorithm}
\caption{Finite Active-Set Search for Exact Projection}
\label{alg:active_search}
\begin{algorithmic}[1]
\Require Nominal input $u_{\rm nom}$, affine inequality data $A$, $b$
\Ensure Exact projection $u^\ast$

\If{$Au_{\rm nom}\le b$}
\State Set $u^\ast = u_{\rm nom}$.
\Else
\For{each $\mathcal I \subseteq {1,\ldots,q}$ satisfying $|\mathcal I|\le \min({m,q})$}
\If{$A_{\mathcal I}$ has full row rank}
\State Compute
\begin{align}
u_{\mathcal I}
=
u_{\rm nom}
-
A_{\mathcal I}^T
\left(
A_{\mathcal I}A_{\mathcal I}^T
\right)^{-1}
\left(
A_{\mathcal I}u_{\rm nom}
-
b_{\mathcal I}
\right). \nn
\end{align}
\State Compute
\begin{align}
\lambda_{\mathcal I}
=
\left(
A_{\mathcal I}A_{\mathcal I}^T
\right)^{-1}
\left(
A_{\mathcal I}u_{\rm nom}
-
b_{\mathcal I}
\right). \nn
\end{align}
\If{$Au_{\mathcal I}\le b$ and $\lambda_{\mathcal I}\ge 0$}
\State Set $u^\ast = u_{\mathcal I}$.
\State \Return $u^\ast$.
\EndIf
\EndIf
\EndFor
\EndIf
\end{algorithmic}
\end{algorithm}


Note that the number of candidate active sets considered by Algorithm~\ref{alg:active_search} is
\begin{align}
    N_{\rm sets}
        =
            \sum_{\ell=0}^{\min\{m,q\}}
            {q \choose \ell}.
\end{align}
This count includes the empty set, which corresponds to the case in which
$u_{\rm nom}$ is already feasible.  For example, if $m=2$, then
\begin{align}
    N_{\rm sets}
        =
            1+q+\frac{q(q-1)}{2}.
\end{align}
If $m=3$, then
\begin{align}
    N_{\rm sets}
        =
            1+q+\frac{q(q-1)}{2}
            +
            \frac{q(q-1)(q-2)}{6}.
\end{align}

Thus, when the violated-set correction fails its certificate, the exact projection is recovered by enumerating candidate active sets and applying the same primal-dual certificate to each candidate.  
Any candidate that satisfies the certificate is an exact solution of the Euclidean projection problem.

Note that if more than one candidate active set satisfies the certificate, this does not imply that the Euclidean projection is nonunique.  
Since the feasible set $\SU,$ given by \eqref{eq:feasibleSet}, is convex and the objective function $\frac{1}{2}\|u-u_{\rm nom}\|_2^2$ is strictly convex in $u$, the projection problem has a unique minimizer.  
Therefore, all certified candidates correspond to the same projected input.

Finally, note that this algebraic active-set search recovers the exact CBF-QP solution without invoking a generic numerical optimizer, while retaining the certification logic of Algorithm~\ref{alg:optimization_free_cbf_filter}.

\section{Numerical Simulation}
\label{sec:numerical_simulation}

This section illustrates the exactness certificate and finite active-set recovery mechanism using a two-dimensional double-integrator system. 
The example is not intended to demonstrate a new obstacle-avoidance controller, but rather to isolate the projection mechanism introduced in this paper. 
In particular, this example focuses on a two-obstacle case, where multiple affine CBF constraints interact and the violated-set correction can fail the exactness certificate.

Consider the system
\begin{align}
    \dot q &= v, \\
    \dot v &= u,
    \label{eq:double_integrator_sim}
\end{align}
where \(q,v,u\in\mathbb{R}^2\).
The objective is to drive the system from an initial position to a desired goal position while avoiding circular obstacles.

The nominal control signal $u_{\rm nom}$ is generated by the PID controller 
\begin{align} 
    u_{\rm nom} &= K_p e + K_i \int e - K_d v,
\end{align}
where $e \isdef q_{\rm g}-q.$
Since the desired goal position is constant, the term $-K_dv$ provides derivative feedback. 
In all simulations, the gains are set as 
$K_p = 0.8 I_2,
K_i = 0.03 I_2, $ and $
K_d = 1.7 I_2.$

For each circular obstacle with center $c_j$ and radius $R_j,$ define the control barrier function
\begin{align}
    h_j(q)
        =
            \|q-c_j\|^2-R_j^2.
\end{align}
Since $h_j$ has relative degree two with respect to the input $u$, we use the high-order CBF condition
\begin{align}
    \ddot h_j+\alpha_1\dot h_j+\alpha_0 h_j \ge 0.
\end{align}
For the double-integrator dynamics,
\begin{align}
    \dot h_j
        &=
            2(q-c_j)^T v,\\
    \ddot h_j
        &=
            2v^Tv+2(q-c_j)^T u.
\end{align}
Therefore, each obstacle induces the affine inequality
\begin{align}
    -2(q-c_j)^Tu
        \le
            2v^Tv
            +
            \alpha_1\dot h_j
            +
            \alpha_0 h_j.
\end{align}
Equivalently, in the notation of this paper,
\begin{align}
    A_j(x)
        &=
            -2(q-c_j)^T,
    \label{eq:exmpAx}        
    \\
    b_j(x)
        &=
            2v^Tv
            +
            2\alpha_1(q-c_j)^Tv
            +
            \alpha_0\left(\|q-c_j\|^2-R_j^2\right),
    \label{eq:exmpbx}        
\end{align}
where $x \isdef \matl q \\ v \matr.$
In all simulations, we set $\alpha_0=\alpha_1=4.$

The exactness certificate is not guaranteed to hold when multiple affine CBF constraints are present. 
The purpose of this example is to illustrate this limitation, identify the mechanism responsible for the failure, and
show why the certification and exact-recovery steps in Algorithm~\ref{alg:optimization_free_cbf_filter} are necessary.

The system is simulated using a fixed time step of \(0.005\) s over a \(20\) s interval.
The goal is
\begin{align}
    q_{\rm g}
        =
            \begin{bmatrix}
                10 \\ 0
            \end{bmatrix}.
\end{align}
The two circular obstacles are centered at
\begin{align}
    c_1
        =
            \begin{bmatrix}
                3 \\ 0.5
            \end{bmatrix},
    \qquad
    c_2
        =
            \begin{bmatrix}
                7 \\ -1
            \end{bmatrix},
\end{align}
and both have radius $R_1=R_2=1.25.$
The two obstacle-induced affine CBF inequalities are obtained from \eqref{eq:exmpAx}, \eqref{eq:exmpbx} by setting $(c_j,R_j)=(c_1,R_1)$ and $(c_2,R_2)$, respectively.

Figure~\ref{fig:two_obstacle_trajectory} shows the closed-loop trajectory.
The filtered trajectory remains outside both obstacles and converges to the goal.
The figure provides the physical context for the example, where the safety filter enforces the obstacle-avoidance constraints while the nominal controller drives the system toward the goal.

\begin{figure}[h]
    \centering
    \includegraphics[width=1\columnwidth]{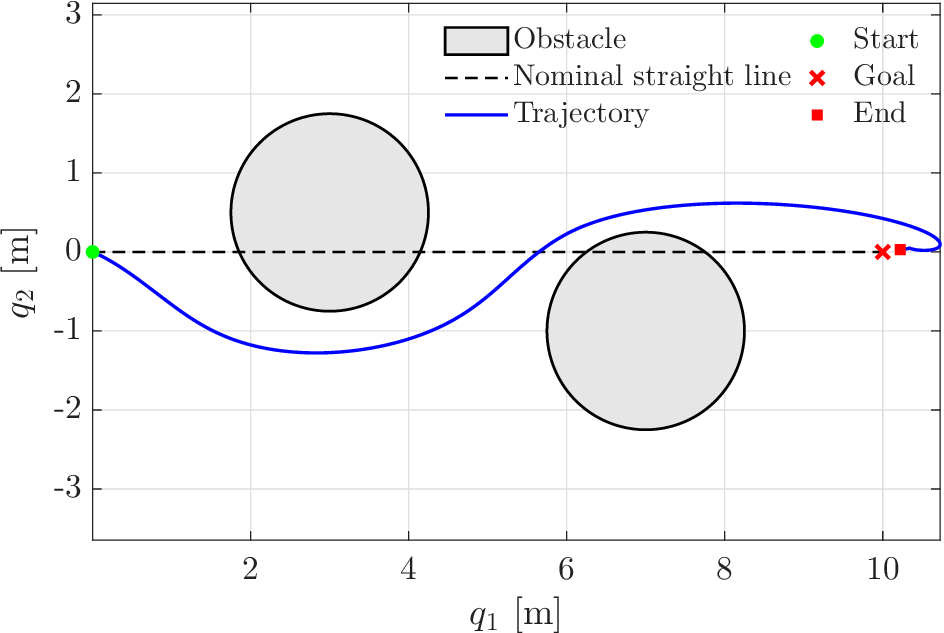}
    \caption{Closed-loop trajectory of the double-integrator system with two circular obstacles.  
    The trajectory remains outside both unsafe sets and reaches the goal.}
    \label{fig:two_obstacle_trajectory}
\end{figure}

Figure~\ref{fig:two_obstacle_filter_status} shows the obstacle barrier functions and the filter status.
Both barrier functions remain nonnegative, which confirms that the closed-loop trajectory remains in the safe set.
The filter status shows three cases: the nominal input is accepted when it already satisfies the CBF inequalities, the violated-set correction is applied when it is certified as exact, and the finite active-set search in Algorithm~\ref{alg:active_search} is used when the certificate fails.
Thus, with multiple affine CBF constraints, the violated-set correction is not always certified as the exact Euclidean projection.

\begin{figure}[h]
    \centering
    \includegraphics[width=1\columnwidth]{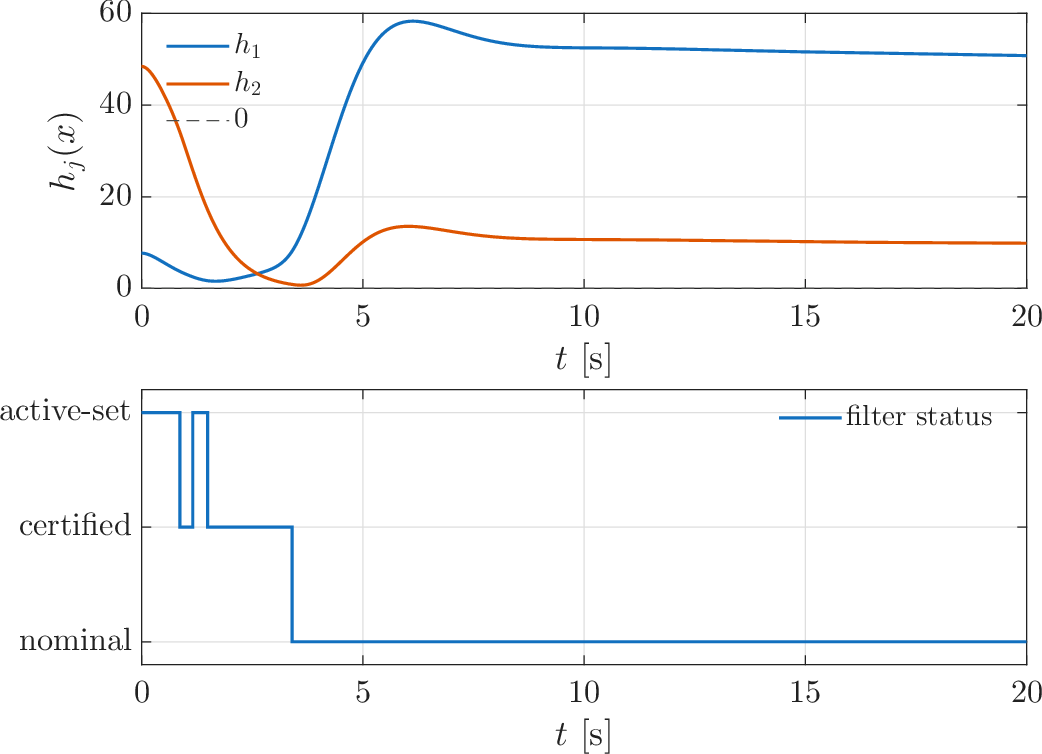}
    \caption{Barrier functions and filter status for the two-obstacle example.
            The lower panel shows the filter status, where nominal input accepted, certified closed-form correction, and finite active-set search are encoded as distinct status values.}
    \label{fig:two_obstacle_filter_status}
\end{figure}

To explain this failure, we evaluate the exactness certificate from
Theorem~\ref{thm:exactness_characterization} along the simulated trajectory.
At each time step, the violated-set candidate $u_{\SV}$ and multiplier
candidate $\lambda_{\SV}$ are computed using
\eqref{eq:violated_projection} and
\eqref{eq:exactness_dual_feasibility}, respectively.

Figure~\ref{fig:two_obstacle_certificate} shows the residuals $r_j$, the
multiplier candidate $\lambda_{\SV}$, and the maximum primal violation
\begin{align}
    \max_i\left(A_i u_{\SV}-b_i\right).
\end{align}
During the interval in which Algorithm~\ref{alg:active_search} is used, the maximum primal violation is nonpositive. 
Thus, the violated-set candidate is feasible with respect to all affine CBF inequalities. 
However, one component of the multiplier candidate is negative. 
Therefore, the KKT conditions for the Euclidean projection fail because of dual infeasibility, not primal infeasibility.

\begin{figure}[t]
    \centering
    \includegraphics[width=1\columnwidth]{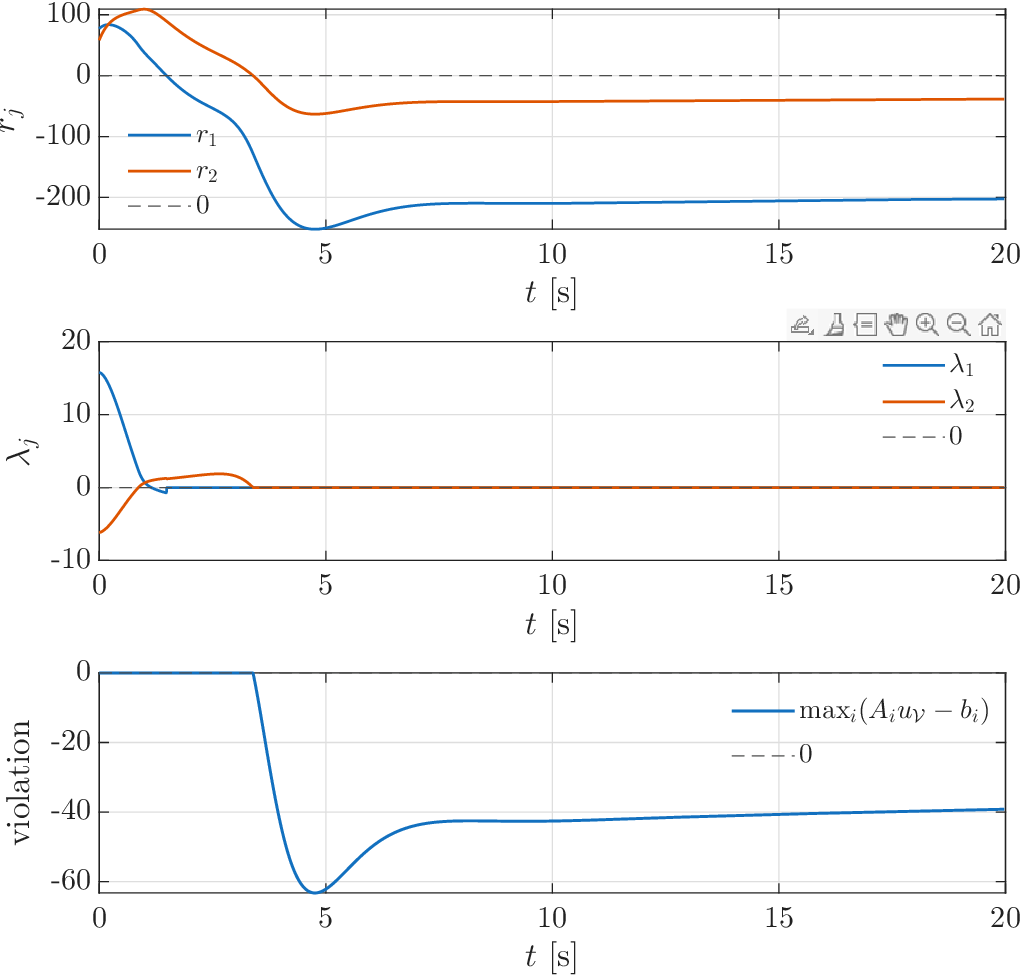}
    \caption{Exactness-certificate diagnostics for the two-obstacle example.
    The violated-set candidate remains primal feasible during the interval in
    which the finite active-set search is used, as indicated by the
    nonpositive maximum primal violation. However, the multiplier candidate has
    a negative component. Hence, the violated-set correction is feasible but is
    not the exact Euclidean projection.}
    \label{fig:two_obstacle_certificate}
\end{figure}

This example highlights an important limitation of extending the
single-constraint intuition to multiple simultaneously violated constraints.
With a single violated affine inequality, the multiplier candidate is a
nonnegative scalar whenever the residual is positive. With two simultaneously
violated constraints, however, exactness depends on the full multiplier vector
\begin{align}
    \lambda_{\SV}
        =
            \left(
                A_{\SV}A_{\SV}^T
            \right)^{-1}
            r_{\SV}.
\end{align}
Even when $r_{\SV}$ is componentwise positive, multiplication by
$\left(A_{\SV}A_{\SV}^T\right)^{-1}$ can produce a multiplier vector with a
negative component. In that case, the violated-set correction may be a safe
control input, but it is not the closest safe input to $u_{\rm nom}$. 
The finite active-set search in Section~\ref{sec:active_set_search}, summarized in
Algorithm~\ref{alg:active_search}, is therefore necessary to recover the
exact CBF-QP solution when the violated-set correction fails the certificate.

The certified closed-form filter was also compared with a standard CBF-QP solver evaluated at the same states and nominal inputs. 
Table~\ref{tab:timing_comparison} shows that the certified filter reduced the average computation time per
filter call by a factor of approximately $38.7$.
The maximum difference between the two filtered inputs over the simulation was
$$
\max_k
\|u_{{\rm cert},k}-u_{{\rm qp},k}\|
=
9.76\times 10^{-8},
$$
and the mean difference was $4.31\times 10^{-9}$. 
Thus, in this example, the certified filter recovers the CBF-QP solution to numerical precision while substantially reducing computation time.

\begin{table}[h]
\centering
\small
\caption{Timing comparison with a standard CBF-QP solver.}
\label{tab:timing_comparison}
\begin{tabular}{lccc}
\hline
Metric & Certified filter & CBF-QP solver & Speedup \\
\hline
Total time [s]      & $6.81\times 10^{-2}$ & $2.63$              & $38.6\times$ \\
Mean time/step [s] & $1.70\times 10^{-5}$ & $6.58\times 10^{-4}$ & $38.7\times$ \\
Max time/step [s]  & $3.81\times 10^{-3}$ & $2.01\times 10^{-2}$ & $5.3\times$ \\
\hline
\end{tabular}
\end{table}

Note that this comparison is intended only to illustrate the computational benefit in this low-dimensional example; the exact speedup depends on the QP solver implementation and hardware.

\section{Conclusion}
\label{sec:conclusion}

This letter characterized when the Euclidean projection associated with a CBF safety filter can be recovered from the affine inequalities violated by a nominal control input.
The violated-set correction enforces the currently violated affine inequalities with equality and can be computed in closed form, but it need not coincide with the exact projection onto the full
feasible set.
A necessary and sufficient exactness certificate was derived in terms of primal feasibility and dual feasibility. 
A structural sufficient condition was then given in terms of interactions among affine-inequality normals, together with geometric special cases that clarify when the violated-set correction is guaranteed to be exact.

The resulting online safety filter first applies the violated-set correction and then certifies whether it is the exact CBF-QP solution. 
When the certificate fails, the exact projection can still be recovered using the finite active-set search proposed in this letter. 
Numerical simulations demonstrated that the violated-set correction can remain feasible while failing to be the exact projection because of dual infeasibility. 
The simulations also showed that the certified closed-form filter with finite active-set search recovered the standard CBF-QP solution to numerical precision while reducing computation time in the two-obstacle example.


\bibliographystyle{elsarticle-num}
\bibliography{cbf_bib}

\end{document}

%% file: PackagesCommands.tex
\usepackage{amsmath} 
\usepackage{amsfonts}
\usepackage{amsthm}

\newtheorem{theorem}{Theorem}

\newtheorem{corollary}{Corollary}[section]
\theoremstyle{definition}

\usepackage{float} 
\usepackage[skip=0em,font=footnotesize]{caption}

\usepackage{tikz}
\usetikzlibrary{shapes,arrows,calc,positioning}
\tikzstyle{bigblock} = [draw, fill=blue!20, rectangle, 
    minimum height=6em, minimum width=8em]
\tikzstyle{medblock} = [draw, fill=blue!20, rectangle, 
    minimum height=4em, minimum width=4em]    
\tikzstyle{mux} = [draw, fill=black!20, rectangle, 
    minimum height=5em, minimum width=0.1em]    
\tikzstyle{smallblock} = [draw, fill=blue!20, rectangle, 
    minimum height=2em, minimum width=3em]
    
\tikzstyle{data_block} = [draw, fill=green!20, rectangle, 
    minimum height=2em, minimum width=3em]
\tikzstyle{ops_block} = [draw, fill=blue!20, rectangle, 
    minimum height=2em, minimum width=3em]    
\tikzstyle{est_block} = [draw, fill=red!20, rectangle, 
    minimum height=2em, minimum width=3em]    
    
\tikzstyle{sum} = [draw, fill=blue!20, circle, node distance=1cm,minimum height=0.5cm]
\tikzstyle{signal} = [coordinate]
\tikzstyle{pinstyle} = [pin edge={to-,thin,black}]
\tikzstyle{block} = [draw, fill=blue!20, rectangle, 
    minimum height=3em, minimum width=9em]
\tikzstyle{blockS} = [draw, fill=blue!20, rectangle, 
    minimum height=3em, minimum width=4em]    
\tikzstyle{input} = [coordinate]
\tikzstyle{output} = [coordinate]
\usetikzlibrary{matrix}

\usetikzlibrary{positioning}
\usetikzlibrary{math}

\usepackage{hyperref}
\usepackage{xcolor}
\hypersetup{
    colorlinks,
    linkcolor={blue!100!black},
    citecolor={blue!50!black},
    urlcolor={blue!80!black}
}

\newcommand{\bc}{\begin{center}}
\newcommand{\ec}{\end{center}}
\newcommand{\benum}{\begin{enumerate}}
\newcommand{\eenum}{\end{enumerate}}
\newcommand{\nn}{\nonumber}
\newcommand{\matl}{\left[ \begin{array}}
\newcommand{\matr}{\end{array} \right]}

\renewcommand{\matl}{\begin{bmatrix}}
\renewcommand{\matr}{\end{bmatrix}}

\newcommand{\matls}{\left[ \begin{smallmatrix}}
\newcommand{\matrs}{\end{smallmatrix} \right]}
\newcommand{\isdef}{\stackrel{\triangle}{=}}

\newcommand{\BBR}{{\mathbb R}}

\newcommand{\real}{{\mathbb{R}}}

\newcommand{\SU}{{\mathcal U}}
\newcommand{\SV}{{\mathcal V}}

\DeclareMathOperator*{\argmin}{arg\,min}

\newenvironment{nonumberproof}
{\begin{proof}\notag}
{\end{proof}}